\documentclass{article}[12pt]
\usepackage{amsmath}
\usepackage[normalem]{ulem}
\usepackage{amsfonts}

\usepackage{color}
   \newtheorem{theorem}{Theorem}[section]
   \newtheorem{lemma}[theorem]{Lemma}
   \newtheorem{proposition}[theorem]{Proposition}
   \newtheorem{corollary}[theorem]{Corollary}

\newtheorem{remark}[theorem]{Remark}

   \newenvironment{proof}[1][Proof]{\begin{trivlist}
   \item[\hskip \labelsep {\bfseries #1}]}{\end{trivlist}}

   \newcommand{\qed}{\nobreak \ifvmode \relax \else
      \ifdim\lastskip<1.5em \hskip-\lastskip
      \hskip1.5em plus0em minus0.5em \fi \nobreak
      \vrule height0.75em width0.5em depth0.25em\fi}

   \numberwithin{equation}{section}

\newcommand{\be}{\begin{equation}} 
\newcommand{\ee}{\end{equation}}
\newcommand{\beq}{\begin{eqnarray}}
\newcommand{\eeq}{\end{eqnarray}}
\newcommand{\bt}{\begin{theorem}}
\newcommand{\et}{\end{theorem}}
\newcommand{\bl}{\begin{lemma}}
\newcommand{\el}{\end{lemma}}
\newcommand{\bc}{\begin{corollary}}
\newcommand{\ec}{\end{corollary}}
\newcommand{\bp}{\begin{prop}}
\newcommand{\ep}{\end{prop}}
\newcommand{\ba}{\begin{array}}
\newcommand{\ea}{\end{array}}

\newcommand{\bbN}{\mathbb{N}}        
\newcommand{\bbZ}{\mathbb{Z}}   
\newcommand{\bbC}{\mathbb{C}}   
\newcommand{\bbt}{\mathbb{T}}   
\newcommand{\bbT}{\mathbb{T}}
\newcommand{\bbr}{\mathbb{R}}   
\newcommand{\bbR}{\mathbb{R}}   
\newcommand{\bbQ}{\mathbb{Q}}   
\newcommand{\bbI}{\mathbb{I}}   

\newcommand{\de}{\mathrm{d}} 

\newcommand{\PLg}{\textrm{PL}_\gamma}

\newcommand{\slr}{  \textrm{ {\bf SL}}_2 (\mathbb{R})  }

\newcommand{\ham}{{\mathfrak h}}

\newcommand{\ly}{\mathcal{L}}

\begin{document}

\title{Dynamical bounds for quasiperiodic Schr\"{o}dinger operators with rough potentials}
\author{Svetlana Jitomirskaya\thanks{Department of Mathematics,
    University of California, Irvine CA, 92717, szhitomi@uci.edu} and Rajinder Mavi\thanks{Department of Mathematics, MSU, East Lansing, MI, 48824}}
\date{}
\maketitle

\begin{abstract}
We establish localization type dynamical bounds as a corollary of
positive Lyapunov exponents for general operators with
one-frequency quasiperiodic potentials defined by piecewise H\"older
functions. This, in particular, extends some results previously known only
for trigonometric polynomials \cite{DT2007} to the case of
surprisingly low regularity. On the technical level, an important part
of the argument is an extension of uniform
uppersemicontinuity to cocycles with discontinuities, a result of
independent interest.
\end{abstract} 
\section{Introduction}
   We will study the quantum dynamical properties of  Schr\"odinger Hamiltonian
   acting on $\ell^2(\bbZ)$.
   \begin{equation}\label{schrodinger1}
    \ham_\theta u(n) = u({n-1}) + u({n+1}) + f(n\omega + \theta) u(n).
   \end{equation}
   where  $\theta \in \bbR$,  $\omega\in\bbR\backslash\bbQ$ and $f:\bbT\to\bbR,$ 
   $\bbT  =  \bbR/\bbZ,$ in the regime of positive Lyapunov exponents.
   The evolution of a wave packet  under the
   Hamiltonian (\ref{schrodinger1}) is given by the formula
   \[ u(t) = e^{-it\ham_\theta }u(0) \]
   Dynamical localization, i.e. the nonspread as $t\to\infty$ of $u(t)$ with
   initially localized  $u(0),$ is related to various quantities that
   can be measured in an experiment. It is often assumed by physicists
   to be a corollary of positivity of Lyapunov exponents, a quantity
   defined by dynamics of the associated cocycle and easily computable
   numerically. As mathematicians, we know however, that positive
   Lyapunov exponents, while implying no absolutely continuous spectrum, can coexist even with almost ballistic transport
   \cite{L96,djls} so one cannot expect dynamical localization in
   full generality, and for a more general result in the direction that
   physicists want, one has to tone down the notion of ``nonspread''
   accordingly. 

For a nonegative function $A(t)$ of time denote
\[  \left\langle A(t)\right\rangle_T=\frac{2}{T} \int_0^\infty e^{-2t/T} A(t) dt \]

 Let
   \[  a(n,t)  =   \frac12\left( \left|\left\langle e^{-it\ham_\theta}
       \delta_0,\delta_n\right\rangle\right|^2 
              +   \left|\left\langle e^{-it\ham_\theta}
       \delta_1,\delta_n\right\rangle\right|^2 
           \right) \]
and
\[  a_T(n)  = \left\langle  a(n,t)\right\rangle_T \]
   Clearly, $\sum_n a_T(n)=\sum_n a(n,t) =1$ for all $t$. 
   The classical quantities of interest are the moments of the
   position operator, that can be defined both with time averaging
    \[ \left\langle |X|_T^p\right\rangle =     \sum_n (1+|n|)^p a_T(n) .\]

or without
\[ \left\langle |X|^p(T)\right\rangle =     \sum_n (1+|n|)^p a(n,T) .\]

   For $p>0$ define the lower and upper transport exponents
   \[  \beta^+(p) =  \limsup_{t\to\infty}\frac{\ln\left\langle |X|^p(t)\right\rangle}{p\ln t};
        \hspace{.1in}  
        \beta^-(p) = \liminf_{t\to\infty}\frac{\ln\left\langle |X|_t^p\right\rangle}{p\ln t}. \]
determining the upper/lower power-law rate of growth of the moments
along subsequences. Note that for the purposes of this paper we define
the upper rate without time averaging, while the lower rate with time averaging. 

Dynamical localization is defined as boundedness in $T$ of
$\left\langle |X|^p(T)\right\rangle $. This implies pure point
spectrum,  
 thus for
parameters for which spectrum is singular continuous (known to be
generic in many situations with positive Lyapunov exponents) one cannot have
dynamical localization in this sense. Then, vanishing of $\beta^+,$
or, in absence of that, at least of $\beta^-$ are properties to look
for. Moreover, dynamical localization is a property that is often unstable with respect
to compact perturbations of the potential or phase shifts. In
contrast, vanishing of $\beta$ is always stable with respect to compact
perturbations \cite{DT2007} and also with respect to phase shifts 
in all known examples. It should be noted however that such vanishing cannot be expected
in general for operators (\ref{schrodinger1}) if the  Lyapunov
exponent is allowed to vanish even on a set of measure zero (as shown
by Sturmian potentials), or, in a slightly more general context even
if it vanishes at a single point \cite{jss}. 

The fact that $\beta(p)$ may depend nontrivially on $p$ is the
signature of intermittency, reflected on an even deeper level
in the fact that different parts of the wave packet may spread at
different rates. While any kind of upper bound discussed above
requires control of the entire wave packet, even control of the spread
of a portion of it is an interesting statement.

Set \[ P(N,t)  = \sum_{|n| \leq N} a(n,t),\,P_T(N) = \sum_{|n|\leq N}a_T(n). \]

A bound of the form 
 $P(T^a,T)>c$ shows that  at time $T$ a portion of the wave packet is
 confined in a box of size $~ T^a$.
 A bound $P_T(T^a) > c$ implies the corresponding statement
 for a weighted average over time. 
  Thus a bound like that holding for
 arbitrary $a>0$ can be considered as a signature of localization.
It is natural in this respect to introduce two other scaling
exponents:
 \[  \overline\xi = \lim_{\delta\to 0} \limsup_{T\to\infty}  \frac{
   \ln(  \inf\{ L |  P_T(L)  > \delta  \}  )  }{  \ln T  } \]
and
 \[  \underline\xi = \lim_{\delta\to 0} \liminf_{T\to\infty}  \frac{
   \ln(  \inf\{ L |  P_T(L)  > \delta  \}  )  }{  \ln T  } \]
Then vanishing of $\overline\xi$ or even $\underline\xi$ is again a
localization-type statement.

Various
   quantities have been used to quantify quantum dynamics, see
   \cite{BGT,DT2007} for a more comprehensive description. In this paper we
   focus on $\xi$ and $\beta$ only. Our main question is what kind of
   localization-type statements can be obtained from positivity of the
   Lyapunov exponents under very mild restrictions on regularity of
   the potential.

While the last decade has seen an explosion of general results for operators  (\ref{schrodinger1}) with analytic $f$ , see e.g. 
\cite{bbook,sim60} and references therein, and by now even the global
theory of such operators is well developed \cite{Aglobal}, there are
very few results beyond the analytic category that do not require
energy exclusion
\footnote{\label{1} It should be noted that exclusion of any energies in
  localization type results, such as, e.g. \cite{cgs}, make upgrading to dynamical statements
  very problematic, as even a single energy that does not carry any
  spectral measure can lead to robust transport \cite{jss,js}} (with few recent
exceptions \cite{wz1,wz2,JMv12} only confirming the rule). Indeed, not only the
methods of proof usually require
analyticity (or at least the Gevrey condition), but certain results
fail to hold as long as analyticity is relaxed \cite{wy} (see also
\cite{JM2}). It is expected that many recent ``analytic'' results
in fact do require analyticity. In this paper we show that, in
contrast to the above, dynamical upper bounds can be obtained as a
corollary of positive Lyapunov exponents under surprisingly weak
regularity.  

Namely, we allow $f$ in (\ref{schrodinger1})  with only H\"older
continuity, and even allow it to have finitely many {\it
  discontinuities} (so only require it to be locally
H\"older). Allowing for discontinuities in the class of considered potentials is
important for two reasons. First, the main
explicit non-analytic operators (\ref{schrodinger1}) that appear in
different contexts in physics literature \cite{p1,p2} have
$f$ with discontinuities.
Several models that are well studied mathematically: Maryland,
Fibonacci (or, more generally, Sturmian)\footnote{It should be
  mentioned that our analysis is not relevant to Fibonacci and most
  Sturmian models as for them the Lyapunov exponent vanishes on the
  spectrum} operators also belong to this class. Second, while there are
few results on positivity of Lyapunov exponents for non-analytic $f$, the Lyapunov exponents of operators
(\ref{schrodinger1}) with discontinuous $f$ are {\it always} positive at
  least a.e. \cite{dk}, providing us with a large collection of models
  for which our results are directly applicable. As far as we know,
  the present paper is the first one holding for a class of potentials that
  rough. Spectral localization for (continuous) H\"older potentials outside a set
  of energies of measure zero was established in \cite{cgs}, but there
  have been no dynamical bounds (see Footnote \ref{1}). In \cite{JMv12}
  we proved continuity of measure of the spectrum for (continuous)
  H\"older potentials.

We will say that $f$ is piecewise H\"older if $f$ has
a finite set of discontinuities, $J_f,$ and  there exists $\gamma>0$
such that $\|f\|_{PL_\gamma} <\infty$ where
   
   \[  \|f\|_{PL_\gamma} = \|f\|_\infty +
     \sup_{h>0} \sup_{t\in\bbT; dist(t,J_f) >|h|}\frac{|f(t+h) - f(t)|}{|h|^\gamma}. \]
 The functions $f$ with finite $\|\cdot\|_{PL_\gamma}$ norm form the
 space of piecewise $\gamma$-Lipschitz functions, that we denote $PL_\gamma(\bbT).$

We will now introduce the Lyapunov exponent.
 For a given $z\in\bbC$, a formal solution $u$ of 
\be \label{schrodinger2} \ham_\theta u = zu\ee
 with operator $\ham_\theta$ given by (\ref{schrodinger1})  
  can be reconstructed from its values at two consecutive points with the transfer matrix
   \begin{equation}\label{evop}
     A^{f,z}(\theta) = 
     \left(\begin{array}{cc}z - f(\theta) & -1\\ 1&0 \end{array} \right); \hspace{.5in} A^{f,z}: \bbt \to SL_2(\mathbb{C})
   \end{equation}
   via the equation
   \begin{equation}\label{dynsol}
      \left(\begin{array}{c}u(n+1)\\ u(n)\end{array}\right) =
     A^{f,z}(\theta+n\omega)      \left(\begin{array}{c}u(n)\\ u(n-1) \end{array} \right).
   \end{equation}
   Let us define the map
     $R:\bbt\to\bbt$ by $Rx:=x+\omega,$ 
   then the pair $(\omega, A^{f,z})$ viewed as a linear skew-product 
   $(x,v)\to(Rx,A^{f,z}(x)v),\;x\in\bbt,\;v\in\bbr^2,$ is called the corresponding Schr\"{o}dinger cocycle.
   The  iterations of the cocycle $(\omega,A^{f,z})$ for
   $k\geq 0$ are given by
   \begin{equation}\label{cocycle}
          A^{f,z}_{k}(\theta) = A^{f,z}(R^{(k-1)}\theta) \cdots A^{f,z}(R^{ 1}\theta) A^{f,z}(\theta),
      \hspace{.2in} A_0^{f,z} = I
   \end{equation} 
   and
      \begin{equation}\label{backcocycle}
      A^{f,z}_{k}(\theta) = \left(A^{f,z}_{-k}(R^{k+1}\theta)\right)^{-1};\hspace{.1in}k<0.
   \end{equation} 
   Therefore, it can be seen from (\ref{dynsol})
   that a solution to (\ref{schrodinger2}) for chosen initial conditions
   $(u(0),u(-1))$ for all $k\in \bbZ$ is given by,
   \begin{equation}
      \left(\begin{array}{c}u(k)\\ u(k-1)\end{array}\right) =  
     \ A_k^{f,z}(\theta)\left(\begin{array}{c}u(0)\\ u(-1)\end{array}\right).
   \end{equation}
   By the  general properties of subadditive ergodic cocycles, we can define the Lyapunov
   exponent
   \begin{equation}\label{avginf}
      \ly(z) = \lim_k\frac{1}{k}\int\ln \|A^{f,z}_k(\theta)\|\de\theta
      = \inf_k \frac{1}{k}\int\ln \|A_k^{f,z}(\theta)\|\de\theta,
   \end{equation}
   furthermore,
   $\ly(z) =\lim_k\frac{1}{k}\ln \|A^{f,z}_k(\theta)\|$
   for almost all $\theta\in\bbt$.
  
Finally, we introduce the Diophantine condition. Writing $\omega$ in
the continued fraction form
      \[\omega = \cfrac{1}{a_1+\cfrac{1}{a_2+\ddots}} \equiv [a_1,a_2,\ldots], \]
      the truncated continued fractions define the 
      approximants $\frac{p_n}{q_n} = [a_1,a_2,\ldots,a_n]$.
  We say that $\omega$ is  Diophantine if for some $\kappa>0$
     \begin{equation}\label{DC} q_{n+1} < q_n^{1+\kappa}\end{equation} for all large $n$.

Our first result is that  just  positivity of the Lyapunov exponent on a
positive measure subset of the spectrum already implies localization
bounds for the transport of the bulk of a wave packet.

Let $\mu_\theta$ be the spectral measure of $\ham_\theta$ and vector
$\delta_0$,
  by which we mean 
  $\langle (\ham_\theta - z)^{-1} \delta_0,\delta_0\rangle = \int_\bbR \frac{d\mu_\theta(x)}{x-z}  $
  for $z$ in the upper half plane and let
 $N:=\int \mu_\theta d\theta$ be the  integrated density of states measure.
   
   \begin{theorem}\label{partloc}For piecewise H\"older $f$ and
     $\omega\in\bbr\backslash\bbQ,$  suppose 
    $\ly(E)$ 
    of (\ref{schrodinger1})  is positive on a Borel subset $U$ with $N(U)>0$.
      Then 
\begin{enumerate}\item For any irrational $\omega$, $\underline\xi=0$ for a.e. $\theta$
\item If $\omega$ is Diophantine, then $\overline\xi=0$ for a.e. $\theta$
\item For all $\theta$ for $\zeta>0$, $P_{T_k}(T_k^{\zeta}) >  C\mu_\theta(U)$ for a 
      sequence $T_k\to\infty$; moreover if $\omega$ is Diophantine, 
      $P_T(T^{\zeta}) >  C\mu_\theta(U) $ for all large $T$. 
\end{enumerate}

   \end{theorem}
\begin{remark}  Positivity of the Lyapunov exponent on a positive
  IDS measure subset is clearly essential, as the result does not
   hold for Fibonacci-type models where Lyapunov exponent is positive
   a.e. but zero on the spectrum.

  The Diophantine condition is  essential for vanishing of
  $\overline\xi$ as the result does not hold for Liouville $\omega$
  (\cite{jz})

\end{remark}
\begin{remark} The full measure sets of $\theta$ in
   cases 1 and 2 of Theorem \ref{partloc}  are specified by the set 
  $\{\theta : \mu_\theta (U) + \mu_{R\theta}(U)>0\}$. It is not entirely clear
  whether there are quasiperiodic examples with $N(E: \ly(E)>0)>0$ and
  $\mu_\theta (E: \ly(E)>0)=0$ for some $\theta$.
\end{remark}
\begin{remark}\label{rem}
It is an interesting question whether or not a.e. vanishing of
$\underline\xi$ is a general corollary of positive Lyapunov exponents,
so holds  for all ergodic potentials. This may be reminiscent of the property of
zero Hausdorff dimension of spectral measures of operators with positive Lyapunov exponents, which was
originally proved for quasiperiodic operators with trigonometric
polynomial potentials \cite{JL00}, but then turned out to be a general
fact, easily extractable from some deep results of potential theory \cite{Si07}.
\end{remark}
\begin{corollary}
Assume $f$ is locally H\"older, has at least one point of
discontinuity and that $N$ has an absolutely continuous component. Then the conclusions of Theorem \ref{partloc} hold.
\end{corollary}
\begin{proof} Follows immediately from a.e. positivity of Lyapunov
  exponents of potentials with discontinuities (was proved in \cite{dk}
  with a conjecture made in \cite{jm}).
\end{proof}

   \begin{remark}
 A large class of examples of operators(\ref{schrodinger1}) with discontinuous
$f$ and absolutely continuous $N$ is presented in \cite{jk}
\end{remark}

Other potentials have been shown to satisfy the conditions of Theorem
\ref{partloc} in various regimes in
\cite{bjerklov,C08,K03,zhenghe}. In all those cases Theorem \ref{partloc}
improves on some of the known results since, even if the results
potentially allowed
for dynamical extensions, unlike Theorem \ref{partloc} spectral
localization cannot hold
for all $\theta$ at least for continuous $f$ that are even on the
hull \cite{js}.

Certainly, not every $f$ in  (\ref{schrodinger1})  corresponds to a
model relevant to physics, and since our main question is physically
motivated, it is natural to impose assumptions that are necessary for
physics relevance. In particular, Lyapunov exponent should be
continuous in various parameters for operators coming from physics
(although such continuity does not hold universally for operators
(\ref{schrodinger1}) even for $f$ in $C^\infty$ \cite{wy}). Our next
result has this as an assumption.

   \begin{theorem}\label{totalloc}
     For piecewise H\"older $f$ and
     $\omega\in\bbR\backslash\bbQ$  suppose $\ly$ is continuous in $E$
     and $\ly(E) > 0$ for every $E\in\bbR$. 
     Then 
\begin{enumerate}\item $\beta^-_{(\omega,\theta)}(p) = 0$ for all $\theta\in \bbT$, $p>0$;
      \item if $\omega$ is  Diophantine, 
      then $\beta^+_{(\omega,\theta)}(p) = 0$ for all $\theta \in \bbT$, $p>0$.\end{enumerate}
   \end{theorem}
\begin{remark}
It is an interesting question whether or not vanishing of
$\beta^-$ is a general corollary of uniformly positive Lyapunov exponents in the
regime of their continuity,
so for all ergodic potentials. The analogy of Remark \ref{rem} may
also apply. 
\end{remark}
\begin{remark} The Diophantine condition is essential for vanishing of
  $\beta^+$ \cite{jz}.
\end{remark}
 It is sometimes useful to consider (\ref{schrodinger1}) with a scalable potential $f$. Let us introduce,
   \begin{equation}\label{schrodingerlambda}
    \ham_{\theta,\lambda} u(n) = u({n-1}) + u({n+1}) + \lambda  f(n\omega + \theta) u(n).
   \end{equation}
 for a parameter $\lambda > 0$.
\begin{corollary}
If $f$ is $C^2$ with exactly two nondegenerate extrema, and $\omega$
is  Diophantine, then there is some $\lambda(f,\omega)> 0$ so that, for $\lambda>\lambda(f,\omega)$,
$\beta^+_{(\omega,\theta)}(p) = 0$ for all $\theta\in \bbT$, $p>0$.
\end{corollary}

Note that this is the first dynamical bound for $C^2$ potentials.

\begin{Proof} Follows directly from Theorem \ref{totalloc} and the
  results of \cite{wz1}.
\end{Proof}

\begin{corollary}
If $f$ is analytic, then, there is $\lambda(f) >0$ so that for $\lambda>\lambda(f)$ both conclusions of Theorem \ref{totalloc}
hold.
\end{corollary}
\begin{proof} follows from non-perturbative positivity  \cite{SS91} and
  continuity \cite{BJ2001} of the Lyapunov
  exponent for analytic $f.$
\end{proof} 
\begin{remark} The conclusions of Theorem (\ref{totalloc})
 were established in a combination of \cite{DT2007}
  and \cite{DT08} for trigonometric
  polynomial $f.$ The result of \cite{DT2007} allows for a 
  weaker Diophantine condition than ours. Namely it
  holds for $\omega$  such that \begin{equation}\label{wbc}
        \lim_{n\to\infty} \frac{\ln q_{n+1}}{q_n} = 0.
     \end{equation}
Our current proof does not automatically extend to this condition
because of the need to tackle low regularity. A simple modification of the
proof allows to obtain this result for $\omega$
satisfying (\ref{wbc}) and analytic $f$ but not  $f \in C^\gamma.$

\end{remark}

Recall that a function $f: \bbT \to \bbR$ is Gevrey if it is $C^\infty$ and there is some $s \leq 1$ so that
       there exist $M,K <\infty$ so that for all $m\geq 1$
\[                \sup_x |\partial^m f(x)|  < M K^m (m!)^{s}.   \]
We say a Gevrey function $f$ satisfies the tranversality condition if  for all  $x\in \bbT$ there exists an $m \geq 1 $ 
  so that $\partial^m f(x) \neq 0$.
\begin{corollary}
If $f$ is Gevrey with a transversality condition and $\omega$ is Diophantine, then there is some $\lambda(f,\omega) > 0$ so that for $\lambda>\lambda(f,\omega),$ $\beta^+_{(\omega,\theta)}(p) = 0$ for all $\theta \in \bbT$, $p>0$.
\end{corollary}
\begin{proof} Follows from Theorem \ref{totalloc} and the results of
  \cite{K03}
\end{proof}

\begin{remark} $\lambda(f,\omega)$ depends on $\omega$ through its
  Diophantine class.  In \cite{K03} Anderson localization is established for
  all $\theta$ and a.e. $\omega$ in this class (depending on
  $\theta.$)
\end{remark}

Another immediate corollary can be obtained for a class of discontinuous $f$
monotone on the period as considered in \cite{jk}.
 Then,  Anderson localization is established in \cite{jk} for Diophantine $\omega$,
 while continuity (and positivity for large $\lambda$) of the Lyapunov
exponent is established for all $\omega$.  
Theorem \ref{totalloc}
immediately implies in this case vanishing of $\beta^-$ for
 $\lambda$ as above and all
$\theta$ and all $\omega$.

For the proof of Theorem \ref{partloc} we use a criterion from
\cite{KKL}, and to prove Theorem \ref{totalloc} apply the results of
\cite{DT08}. This is done in Section \ref{KL}. To apply those results we need to establish certain lower
bounds on transfer matrices. To obtain this we build on the
technique we introduced in \cite{JMv12}. 
On the technical side, 
our main achievement is
in  both extending the method of \cite{JMv12} to allow discontinuities
and in establishing the underlying
uniform upper bound for uniquely ergodic dynamics to the case of
cocycles with zero measure set of discontinuities. The latter is a
general result that is of independent interest and of the type that
has been crucial in various proofs of localization/regularity in
many recent articles. Our extension has already been
used in \cite{jk} for their spectral localization theorem and in
\cite{jz} for their dimensional analysis of Sturmian potentials.

\section{Key lemmas}   \label{KL}

Both \cite{KKL} and \cite{DT08} (see lemma's \ref{KKLthm} and \ref{DTtheorem}) reduce the dynamical bounds to lower
estimates on the norms of the transfer matrices over controlled
scales. Essentially, one needs to establish polynomial growth of any
order at scales uniform in energy. Positivity of the Lyapunov exponents per se implies such lower
bounds, so it remains to establish uniform control over the scales. This is
done in Proposition \ref{polygrow}.

   For any $\delta \geq 1$ and $1\geq \zeta > 0$ we define, for  $E\in \bbC$ and $T >0$
   \[  \Phi_{\zeta,\delta}(E,T) = \inf
                \left(
                  \min\left\{\frac{\max_{1\leq  j\leq 
		     T^\zeta}\|A_{j}(\theta,z)\|^2}{  T^{\delta} }, \frac{\max_{1\leq  j\leq 
		     T^\zeta}\|A_{- j}(\theta,z)\|^2} { T^\delta}\right\} 
                     \right)  \]
   where the infimum is over all $|z - E| \leq T^{-\zeta} $ and $ \theta\in \bbT$. 
We will establish

  \begin{proposition}\label{polygrow}
      Suppose $f\in PL_\gamma(\bbT),\chi>0$, and suppose $\ly(E) >
      \chi $  for all $E$ in a Borel set 
  $U\subset \bbR.$  
      Then, for any $\delta \geq 1$ and  $1\geq\zeta > 0$ we have,
      there is some $c>0$ and sequence $(T_n)$ so that for every $E \in U$,
      there is some $n_E$ so that, for $n > n_E$
      \begin{equation}\label{maxcocycle1}
         \Phi_{\zeta,\delta}(E,T_n) > c.
      \end{equation}  
     If $\omega$ is Diophantine,  for each $E\in U$ there is $T_E < \infty$ so that for $T > T_E$,
           \begin{equation}\label{maxcocycle2}
         \Phi_{\zeta,\delta}(E,T) > c.
      \end{equation} 
      Finally, if $U$ is compact and $\ly$ is continuous then for all
       rotations $\omega$ there is some $n_0 < \infty$ so that for all $E\in U$ we can set $n_E = n_0$ in
       (\ref{maxcocycle1}),
       and if $\omega $ is Diophantine there is some $T_0 < \infty$ so that for all $E\in U$ 
       we can set $T_E = T_0$ in (\ref{maxcocycle2}).
   \end{proposition}
   This propostion  
is essentially a corollary of the following 
 Lemma. 
    
  \begin{lemma}\label{SOlemma}
     Suppose $f\in\PLg$, $\ly(E)>0.$ 
  For any $\tau>0$  
     there exists $k_\tau  =  k_\tau(E) < \infty$ so that if
     $q_n>e^{k_\tau\ly(E)\tau\over \gamma}$, then for any $k\in\bbZ^+$ such that
     $k_\tau<k<\frac{\gamma}{\ly(E)\tau}\ln q_{n}$
     then for
     any $\theta \in \bbT$
     there is some $0<x\leq q_n + q_{n - 1} - 1$  so that for any $z\in\bbC$ with $|z-E|<\exp\{-\tau k\ly(E)\}$
     \[ \left\|A_k^{f,z}\left(R^x\theta\right)\right\|\geq e^{k(1-\tau)\ly(E)}. \]
     If $\ly$ is continuous and $U$ is compact, then $k_\tau$ can be chosen uniformly over $E \in U$. 
   \end{lemma}
   We will in fact prove a more general statement, for cocycles defined in a neighborhood of $f$, see Lemma \ref{SOlemma2}.
   The proofs of Proposition \ref{polygrow} and Lemma \ref{SOlemma}
   are in section \ref{SA}.  They are based on
    section \ref{UU} where we prove convergence results for discontinuous cocycles in a general setting.
   The remainder of this section is dedicated to proving Theorems \ref{partloc} and \ref{totalloc}.
 
For $f:\bbZ \to H,$ where $H$ is some Banach space and $L\geq1$, the
   truncated $\ell^2$ norm in the positive direction is defined as  
   \[ \| f \|_L^2 = 
           \sum_{n=1}^{\lfloor L\rfloor} |f(n)|^2  
            + \left(L-\lfloor L\rfloor\right) |f(\lfloor L\rfloor +1)|^2. \]
   The truncated $\ell^2$ norm in both directions, for $L_1,L_2 \geq1$, will be denoted
   \[ \|f \|^2_{L_1,L_2} = 
            \sum_{n= - \lfloor L_1\rfloor}^{\lfloor L_2 \rfloor} |f(n)|^2 + 
                 \left(L_1-\lfloor L_1\rfloor\right) |f(-\lfloor L_1\rfloor -1)|^2
                 +\left(L_2-\lfloor L_2\rfloor\right) |f(\lfloor L_2\rfloor +1)|^2. \]
   With $A_\bullet(\theta,E)$ a function on $\bbZ$,
   define $\tilde{L}^+_\epsilon(\theta,E)\in\bbR^+$ by requiring that
   the truncated $\ell^2$ norm obeys
   \[ \| A_\bullet(\theta,E)\|_{\tilde{L}^+_\epsilon(\theta,E)} 
                    =  2\|A_1(\theta,E)^{-1}\|\epsilon^{-1}. \]
   We now recall the following result of Killip, Kiselev and Last,
   \begin{lemma}\label{KKLthm}{(Theorem 1.5 of \cite{KKL})}
      Let $\ham$ be a Schr\"{o}dinger operator and $\mu$ the spectral
      measure of $\ham$ and $\delta_1.$
      Let $T > 0$ and $L_1,L_2 > 2$, then
      \begin{equation}\label{KKLloc}
         \langle \|e^{-it\ham}\delta_1\|^2_{L_1,L_2} \rangle_T 
           > C\mu\left(\left\{ E:\tilde{L}^-_{T^{-1}} \leq L_1;
                            \tilde{L}^+_{T^{-1}} \leq L_2 \right\}\right)
      \end{equation}
      where $C$ is a universal constant.
   \end{lemma}

   {\bf Proof of Theorem \ref{partloc}.}

We first prove part 3. Assume $\mu_\theta(U)>0 $. For
      $\epsilon>0,$ let $\chi>0$ be such that
      \[ \mu_\theta(\{E\in U: \ly(E) > \chi \}) > \mu_\theta(U) -
      \frac{\epsilon}{2}. \]
Let $\zeta >0.$ First consider the Diophantine case. Then by Proposition
\ref{polygrow} with $\delta=3$, for $E\in U$ we have
$\Phi_{\zeta,3}(E,t)  > c_E>0$ for $t>T_E.$ Therefore we can find $
M_\epsilon>0$, so that outside a set of $E$ of measure $\frac{\epsilon}{2}$,

      \begin{equation} \label{growth1}
        \| A_\bullet(\theta,E)\|_{T^\zeta} > T
      \end{equation}
   for $T>M_\epsilon.$ 
      Thus $\tilde{L}_{T^{-1}}^\pm(\theta,E) < T^{\zeta}$ for all $T>M_\epsilon$.
      We have from Lemma \ref{KKLthm}
      \begin{equation}\label{pc1}
         \langle \|e^{-it\ham_\theta}\delta_0\|^2_{T^\zeta} \rangle_T 
                                      > C(\mu_\theta(U) - \epsilon).
      \end{equation}
      If $\omega$ is not Diophantine, (\ref{growth1}) is satisfied for a
      sequence $T_k\to\infty$, thus (\ref{pc1}) holds for a sequence $T_k$.
      As (\ref{pc1}) holds for all $\epsilon$ we can let $\epsilon\to
      0$.\qed

To prove parts 1,2 note that if $\mu_\theta(U) + \mu_{R\theta}(U)>0 $
we have by (\ref{pc1}) that either $\langle
\|e^{-it\ham_\theta}\delta_0\|^2_{T_k^\zeta} \rangle_{T_k} >c$ or $\langle
\|e^{-it\ham_\theta}\delta_1\|^2_{T_k^\zeta} \rangle_{T_k} >c$, so
$P_{T_k}(T_k^\zeta)>c,$ 
thus
 $\frac{   \ln(  \inf\{ L |  P_{T_k}(L)  > \delta  \}  )  }{  \ln T_k } <\zeta.$
 Since $\zeta>0$ is arbitrary,  for such
 $\theta$ we have that  $\underline\xi$ (or $\overline\xi$ for Diophantine $\theta$)
 are equal to zero. Finally observe that since the set $\{\theta: \mu_\theta(U) + \mu_{R\theta}(U)>0\}$ is shift invariant,
      $N(U)>0$ implies $\mu_\theta(U) + \mu_{R\theta}(U)>0 $ for
      a.e. $\theta$. \qed

   The following result of Damanik and Tcheremchantsev allows us to control
   the evolution of the entire wavepacket.
   \begin{lemma}{(Corollary 1 of \cite{DT08} plus Theorem 1 of \cite{DT2007})}\label{DTtheorem}
     Let $\ham_\theta$ be operator (\ref{schrodinger1}), 
     with $f$ real valued and bounded, 
     and $K\geq 4$ is such that $\sigma(\ham_\theta)\subset [-K+1,K-1]$.
     Suppose for all $\zeta\in(0,1)$, we have
     \begin{equation} \label{dtgrowth}
       \int_{-K}^K \left(\min_{ s \in\{-1,1\}}\max_{1\leq n \leq T^\zeta}
          \left \|A_{s n}\left(E+\frac{i}{T}\right)\right\|^2\right)^{-1}\de E
		  = O(T^{-\delta})
     \end{equation}
     for every $\delta \geq 1.$ Then $\beta^+(p) = 0$ for all $p>0$.
     If (\ref{dtgrowth}) is satisfied for a sequence
     $T_k\to\infty,$ then $\beta^-(p) = 0$ for all $p > 0$.
   \end{lemma}

    {\bf Proof of Theorem \ref{totalloc}}
      Assume $\sigma(\ham)\subset [-K+1,K-1]$ . Let $R > K$ and let $\chi = \inf_{|z|<R}\{\ly(z)\}$.
      We assume continuity so $\chi > 0$ and there exists 
      a large $M<\infty$ so that (\ref{maxcocycle2}) holds 
      uniformly for all $T > M$ and
      $E\in\{E\in\bbC:|\mathcal{R}(E)| \leq K;|\mathcal{I}(E)|\leq 1\}$.
      Thus for large enough $T$ and $\omega$ Diophantine we have
      \begin{equation} \label{growth2}
          \int_{-K}^K \left(\max_{1\leq \pm n \leq T^\zeta}
          \left \|A_n\left(E+\frac{i}{T}\right)\right\|^2\right)^{-1}\de E
	  \leq CKT^{-\delta}
	  = O(T^{-\delta}).
      \end{equation}
      If $\omega$ is not Diophantine, we use (\ref{maxcocycle1}) to find
      (\ref{growth2}) is satisfied for a sequence of $T_k\to \infty$.\qed

   \section{Rough cocycles}\label{UU}

   The goal of this section is to establish the uniformity of uppersemicontinuity of the Lyapunov exponent.
   It is known (see e.g. \cite{JMv12}) the pointwise Lyapunov exponent
   has uniform upper bounds in small neighborhoods for continuous cocycles.
   Here we show  the requirement of continuity of cocycles can be relaxed.
   Let $(X,R,\mu)$ be a uniquely ergodic compact Borel probability space.
   We will say a function $f$ is {\it almost continuous} if its 
   set of discontinuities has a closure of measure zero.
   Let $\mathbb{B}_\infty(X)$ be the space of bounded functions on $X$
   with
   \[ \|f\|_\infty = \sup_{x\in X} |f(x)|, \]
   Notice that sets of measure zero are not dismissed by this norm. 
   For a Borel set $D\subset X$ define a seminorm 
   \[ \|f\|_{D,\infty} = \sup_{x\in D} |f(x)|. \]

   A subadditive cocycle on $(X,T,\mu)$ is a sequence of functions $f_1,f_2,\ldots$ on $X$
   so that $f_{n+m}(x) \leq f_n(x) + f_m(R^nx)$.
   We use the notation $\{f\}$ for a subadditive cocycle $f_1, f_2,\ldots.$   
   Let $\Delta(X)$ be the set of all $\{f\}$ with $f_n\in
   \mathbb{B}_\infty$ for all $n$ . By Kingman's subadditive ergodic theorem \cite{W},
   a subadditive cocycle $f_n(\cdot)$ on $(X,T,\mu)$ obeys, 
   for $\mu$-almost all $x\in X$,
   \[ \lim_{n\to\infty}\frac{1}{n}f_n(x)
      = \lim_{n\to\infty}\int_Xf_n(x)\mu(\de x) = \Lambda(f)\] 
   Let $E_n = E_n(\{f\})$ be the closure of the set of discontinuities of $f_n$.
   For a set $E\subset X$ define a ball, $B_{\delta} (E) = \{x\in X: \exists e\in E, |x-e| < \delta\}$.
   Then we introduce, for $\delta\geq 0$ 
   the sequence of sets $D_n = X\backslash \overline{B_{\delta}
     (E_n(\{f\})) }$, and a pseudometric
   \[ 	  \de_\delta\left(\{g\},\{f\}\right) = \sum_{n\geq 1}
           \frac{1}{2^n}\frac{\|g_n - f_n\|_{D_n,\infty}}{1+ \|g_n-f_n\|_{D_n,\infty}}. \]
   From this pseudometric we define the $\delta$-$\sigma$ neighborhood of $\{f\}$ as,
   \[\mathcal{N}_{\delta;\sigma}(\{f\}) = 
         \left\{\{g\}:\de_\delta(\{f\},\{g\})<\sigma\right\}.\]

   \begin{theorem}\label{pwfurman}
      Suppose $\{f\}\in\Delta(X)$ so that $f_n$ is almost continuous for all $n.$ Let $\epsilon > 0$.
      There exists  $\delta > 0$ and $\sigma > 0$ and $K < \infty$ all depending on $\{f\}$ and $\epsilon$
       so that for $\{g\} \in\mathcal{N}_{\delta;\sigma}(\{f\})\cap \mathbb{B}_\infty$
	  and $n >  K$ implies
	  \[\frac{1}{n}g_n < \Lambda(f) + \epsilon\max\left\{\|g\|_\infty,1 \right\}\]
   \end{theorem}
   The result extends the theorem of Furman \cite{F97} and our recent
   extension of it \cite{JMv12} to the case
   of almost continuous subadditive cocycles.
   Here is a simple appplication of the theorem to a single subadditive cocycle.
   \begin{corollary}\label{pwfurmana}
      Suppose $f_n$ are almost continuous and subadditive
      and $\|f_1\|_\infty<\infty$.
      For any $\epsilon> 0$ there is $K <\infty $ so that for 
      $n > K$ and all $x\in X$ we have
      \[ \frac{1}{n}f_n(x) < \Lambda(f) + \epsilon 
\]
   \end{corollary}
   A further corollary arises in the application to 
   matrix cocycles for an almost continuous matrix 
     $M:X\to SL_{2}(\bbC)$.  
  Let $R$ be a uniquely ergodic transformation on $X$, and define the associated cocycle,
    \[    M_k(\theta) = M(R^{(k-1)} \theta) \cdots M ( R \theta) .    \] 
    Let $E$ be the set of discontinuities of $M,$
   let $B = B_\delta(E)$ be the set of points with in distance $\delta$ of $E$.
    Let $E_n=E_n(\{\ln\|M_n\|\})$ 
    and notice that $B_\delta(E_n) \subset \cup_{\ell = 0}^{n-1}  R^{-\ell} B  $.
       Let $D_n = X\backslash \overline{ B_\delta (E_n)  }$.
   Finally,  let $\ly(M)$ be the Lyapunov exponent $\Lambda(\{\ln\|M_n\|\}).$  
   \begin{corollary}\label{mc0}
       Suppose $M:X\to SL_{2}(\bbC)$ is almost continuous and bounded. 
      For any $\epsilon>0$, there is 
      $\delta>0,$ $\rho>0$, and $K<\infty$   such that  
       if $\overline M:X\to SL_{2}(\bbC)$ is
        bounded \footnote{ Note there is no assumption of continuity} so that
      $\|\overline{M} - M\|_{X\backslash B_\delta(E),\infty}<\rho$, then $k > K$ implies
      \[\|M_k(\theta)-\overline{M}_k(\theta)\| <
           \max_{0\leq i\leq k-1}\{\|M(R^i\theta) - \overline{M}(R^i\theta)\|\}
            e^{k(\ly(M) + \epsilon\max\{1,\ln\|\overline M\|_\infty\} )}  \]
   \end{corollary}
   
For our application we only need the $\delta=0$ version:
   
   \begin{corollary}\label{mc}
        Suppose $M:X\to SL_{2}(\bbC)$ is almost continuous and bounded. 
      For any $\epsilon>0$, there is $\rho>0$ and $K<\infty$ such that 
      if $\overline M:X\to SL_{2}(\bbC)$ is
        bounded and 
      $\|\overline{M} - M\|_{\infty}<\rho$, then $k > K$ implies
      \[\|M_k(\theta)-\overline{M}_k(\theta)\| <
           \max_{0\leq i\leq k-1}\{\|M(R^i\theta) - \overline{M}(R^i\theta)\|\}
            e^{k(\ly(M) + \epsilon\max\{1,\ln\|\overline M\|_\infty\} )}  \]
   \end{corollary}

  { \bf Proof of Theorem \ref{pwfurman}} 
 Let $\epsilon < \left(1 + 2\|f_1\|_\infty\right)^{-1}$. 
      $X\backslash E_n$ is an open set of full measure, and for every 
      $x\in X\backslash  E_n$, $f_n$ is continuous in a neighborhood of $x$.
      The set 
      \[ J_n = 
        \left\{ x\in X\backslash  E_n: 
           |\tfrac{1}{n}f_n(x) - \Lambda(f)| < \epsilon \right\} \]
      is open and by Kingman's theorem  $\mu (J_n^c) \to 0$ as $n\to \infty$. 
 
      Let $n>1$ be large enough so that  $\mu(J_n^c) < \epsilon$.
      Let $\delta>0$ be such that $ \mu(\overline{B_\delta(E_n))} < \epsilon$.
      Define $D_n = X\backslash \overline{B_\delta(E_n)}$. 
      For any $\{g\}\in\mathcal{N}_{\delta,\epsilon/2^{n}}(\{f\})$,  
      and for $x\in J_n \cap D_n$ we have $f_n(x) < n(\Lambda(f) + \epsilon) $
      which implies
      \begin{equation}\label{fe1}
		g_n(x) \leq |f_n(x)| + |g_n(x) - f_n(x)|
            <  n(\Lambda(f) + \epsilon) + 2\epsilon \leq n(\Lambda(f) + 2\epsilon).
      \end{equation}

Note that $J_n^c \cup D_n^c$ is a closed set of $\mu$ 
      measure less than $2\epsilon$. We will now follow the idea in the
      Weiss-Katznelson proof of Kingman's theorem \cite{KW}, adapting it to the setting with
      discontinuities.
      By regularity of the Borel measure, there is an open set
      $D$ containing $J_n^c \cup D_n^c$ of measure less than $3\epsilon$,
      and by Urysohn's lemma there is a continuous function $0\leq
      h\leq 1$
      so that $h|_{J_n^c \cup D_n^c} = 1$ and $h|_{D^c} = 0$. Since
      $(X,T,\mu)$ is compact uniquely ergodic there exists some
      $M_1<\infty$ so that for $M>M_1$ and all $x$,
      $|\frac{1}{M}\sum_{i=1}^M h(T^ix) - \int h d\mu|< \epsilon$.
      For any $x\in X$ construct a sequence $(x_i)$ in $X$ in the following way. 
      For $i = 1$ let $x_1 = x$ 
      and for subsequent terms let $x_{i+1} = T^{n_i}x_i$;
      where $n_i$ is defined as
      \[ n_i = n_i(x) =\left\{
       \begin{array}{cl}n, &\textrm{ if } x_i\in J_n\cap D_n \\
             1, & \textrm{ otherwise.}  \end{array}\right.  \]

      We now consider the cocycles for a sufficiently large index.
      Let $M  > \max\{\frac{n}{\epsilon},M_1\}$, and choose $p$ so that
      \[n_1 +\cdots +n_{p-1} \leq M < n_1+\cdots+n_p. \]
      Let $K = M - \left(n_1 +\cdots+n_{p-1}\right)\leq n$.
      By subadditivity,
      \[ g_M(x) \leq 
         \sum^{p-1}_{i=1} g_{n_i}(x_i) + g_K(x_p)
          \leq \sum^{p-1}_{i=1} g_{n_i}(x_i) + n\|g_1\|_\infty. \]
     Partition the above sum into $x_i\in D_n\cap J_n$ and $x_i\in
     D_n^c\cup J_n^c.$ On the former set use (\ref{fe1}) and on the latter use the trivial bound $\|g_1\|_\infty$. 
      \be \label{fe2}   g_M(x) \leq  
        \sum^{p-1}_{i=1}\left[ n_i\left(\Lambda(f) 
                     + 2\epsilon\right){\bf 1}_{ J_n\cap D_n }(x_i)
                     + \|g_1\|_\infty \cdot{\bf 1}_{J_n^c \cup D_n^c}(x_i)\right] + n\|g_1\|_\infty.\ee
      
      Therefore, we have uniformly in $x$,
      \[ \sum^{p-1}_{i=1}
          \|g_1\|_\infty\cdot{\bf 1}_{J_n^c \cup D_n^c}(x_i) \leq
        \sum_{i=1}^M \|g_1\|_\infty\cdot{\bf 1}_{J_n^c \cup D_n^c}(T^ix) 
              \leq \sum_{i=1}^M \|g_1\|_\infty h\left(T^i(x)\right)  < 4\epsilon\|g_1\|_\infty M\]
      Substituting this into the sum on the right hand side of (\ref{fe2}), we find
      \begin{eqnarray}
         \frac{1}{M}g_M(x) &\leq&  \nonumber
         \frac{1}{M}\sum^{p-1}_{i=1} 
               n_i\left(\Lambda(f) + 2\epsilon\right){\bf 1}_{ J_n\cap D_n}(x_i)
             + \frac{1}{M}
               \sum_{i=1}^M \|g_1\|_\infty\cdot{\bf 1}_{ J_n^c\cup D_n^c }(T^ix) 
                + \frac{n}{M}\|g_1\|_\infty \\
             & \leq& \nonumber
               \left(\Lambda(f) + 2\epsilon\right) + 4\epsilon\|g_1\|_\infty + \frac{n}{M}\|g_1\|_\infty \\
             & \leq& \nonumber \Lambda(f) + 2\epsilon +  5\epsilon\|g_1\|_\infty.
      \end{eqnarray}
   \qed
  
   We now prove Corollary \ref{mc0} for almost continuous matrices.
   \begin{proof}.
Set $f_n(x) = \ln\|M_n(x)\|, $ $g_n = \ln\| \overline M_n \|.$ Then,
since $M,\overline{M}$ are in $SL_{2}(\bbC)$, we have 
  \[    \left| \ln\|M_n(x)\| - \ln\|\overline M_n(x)\| \right| 
             \leq  \left| \|M_n(x)\| - \|\overline M_n(x)\|\right| 
             \leq  \|M_n(x) - \overline M_n(x)\|   \] 
Thus for $\delta,\sigma>0$ 
 there exists $\rho>0$ so that $\|M-\overline M\|_\infty <\rho$
 implies $d_\delta(\{f\},\{g\})<\sigma$. Therefore for $\|M-\overline
 M\|_\infty <\rho$ we have
 $\{g\}\in\mathcal{N}_{\delta;\sigma}(\{f\})\cap \mathbb{B}_\infty$
so we are in a
 position to apply Theorem \ref{pwfurman} which then yields that for
 $\epsilon > 0$  there
 exists $n_\epsilon<\infty$ so that 
for $n > n_\epsilon$,  for any $x\in X$, 
    \be\label{furman1}   \|\overline M_n(x) \| < \exp\left\{  n(  \ly + \epsilon Q ) \right\} \ee
where $ Q = \max\{1 ,\ln\|\overline M\|_\infty\} $
and $\ly =\Lambda(\{f\}).$

    We have 
        \[
           \| M_k(\theta) - \overline M_k(\theta) \| \leq
            \sum_{0\leq \ell\leq k-1}
            \| \overline M_\ell(R^{k-\ell} \theta)(\overline M-M)(R^{k-1-\ell}\theta)M_{k-1-\ell}(\theta)\|   \] 

Thus
         \be \label{mcabound}   \| M_k(\theta) - \overline M_k(\theta) \| \leq
        \sup_{0\leq \ell\leq k-1}\left\{  \|(\overline M-M)(R^{k-1-\ell} \theta)\| \right\}
         \sum_{0\leq \ell\leq k-1}
        \| \overline M_\ell (R^{k-\ell}\theta)\| \|M_{k-1-\ell}(\theta)\|.\ee
   Let $k > 2n_\epsilon$. Then we can separate the above sum into $[0,n_\epsilon -1]$, $[n_\epsilon,k -1 -n_\epsilon]$, $[k - n_\epsilon, k-1]$,
   On the second two intervals $\ell \geq n_\epsilon,$ and on the first two intervals $k-1-\ell \geq n_\epsilon$
   so we can apply (\ref{furman1}) to $\overline M_\ell$ and $M_{k-1-\ell}$ respectively. 
   \beq
        \sum_{k-n_\epsilon\leq \ell\leq k-1}
        \| \overline M_\ell (R^{k-\ell}\theta)\| \|M_{k-1-\ell}(\theta)\|
        &\leq&\nonumber
        \sum_{k- n_\epsilon \leq \ell\leq k-1} \|M\|_\infty^{k-1-\ell} \exp\{(k-1-\ell)(\ly + \epsilon Q)  \} \\
        &\leq&\nonumber
        n_\epsilon \|M\|_\infty^{n_\epsilon }  \exp\{(k-1)(\ly + \epsilon Q  )\} 
   \eeq
   Similarly, for $\ell \in [0,n_\epsilon -1]$  
   \beq   \sum_{0\leq \ell\leq n_\epsilon -1}
        \| \overline M_\ell (R^{k-\ell}\theta)\| \|M_{k-1-\ell}(\theta)\|  
                &\leq&\nonumber   \sum_{0\leq \ell\leq n_\epsilon -1} \|\overline M\|_\infty^\ell \exp\{(k-1-\ell)(\ly + \epsilon   Q)\}    \\
               &\leq&\nonumber   n_\epsilon e^{n_\epsilon Q} \exp\{(k-1 
)(\ly + \epsilon   Q\}    
   \eeq
   On the center segment $ \ell \in [n_\epsilon,k -1 -n_\epsilon]$  both cocycles approach the upper Lyapunov limit, 
  so we have using (\ref{furman1})
   \beq
         \sum_{n_\epsilon \leq \ell \leq k-1-n_\epsilon} 
        \| \overline M_\ell (R^{k-\ell}\theta)\| \|M_{k-1-\ell}(\theta)\| & \leq & \nonumber
            \sum_{n_\epsilon \leq \ell \leq k-1-n_\epsilon} \exp\{(k-1)(\ly + \epsilon Q   )\} \\
          &\leq & \nonumber   (k-2n_\epsilon)  \exp\{(k-1)(\ly + \epsilon Q   )\} 
      \eeq 
     Thus, there is some $K < \infty $ so that for $k > K$,
       \[          \sum_{0\leq \ell\leq k-1}
        \| \overline M_\ell (R^{k-\ell}\theta)\| \|M_{k-1-\ell}(\theta)\|   < \exp\left\{  k(  \ly + 2\epsilon Q  ) \right\}   \]
     which together with (\ref{mcabound}) implies the result. \qed
   \end{proof}

  Finally, an immediate corollary is

  \begin{lemma}\label{cuLyapunov}
     For $\ly$ continuous on a compact set $K\subset\bbC$
     given $\epsilon > 0$ there is a $k_\epsilon < \infty$ so that $k > k_\epsilon$
     implies, for $z\in K$ and $\theta\in\bbT$, 
     \[ \|A^{f,z}_k(\theta)\| \leq e^{k(\ly(z) + \epsilon)}.\]
  \end{lemma}
  \begin{proof} Follows immediately by compactness and Corollary \ref{mc}. 
    \qed
  \end{proof}

   \section{Proof of the main Lemmas}\label{SA}

   We will first use Lemma \ref{SOlemma} to obtain Proposition \ref{polygrow}.
   \begin{proof} 
      Fix $f\in\textrm{PL}_\gamma(\bbT)$, $E \in U $,
      $ \delta \geq 1$ and $1 > \gamma > 0$ and $\theta \in \bbT$.
      Let $q_n$ be the sequence of denominators of the continued fraction approximants of $\omega$.
        Boundedness of the Lyapunov exponent on compact sets in $\bbC$
      follows from upper semicontinuity, so we may define
       \[\bar{\chi} = \sup\{\ly(z)\in\bbC:
           |\Re (z)|\leq K;  |\Im (z)|\leq 1 \}.\]
 
      We consider arbitrary irrationals,
      and make a separate argument for the Diophantine case at the end.
      If $\omega$ is Diophantine let $\xi = 1+2\kappa$ where $\kappa > 0$ is
      as described in (\ref{DC}), otherwise, let $\xi  = 1$.
      Let $1>\tau>0$ be such that
      \be \label{PGpars} \frac{\tau}{1-\tau} < \frac{\gamma\zeta\chi}{\delta\xi\bar\chi}  
      \ee
      and choose $\sigma$ so that \[  \frac{\chi \zeta }{\bar\chi\delta\xi}(1-\tau) > \sigma > \tau /\gamma.\]

      Then from Lemma \ref{SOlemma} for $k_\tau(E) < k < \frac{1}{\sigma\bar\chi}\ln q_n$
      there is some $0\leq j \leq q_n + q_{n-1} - 1$ so that for $|E-z| < e^{-\tau\bar\chi k}$
      \be \label{growexp}
       \|A^{f,z}_{k}(\theta+j\omega)\| \geq \exp\{(1-\tau)k\ly\}.  
       \ee
      Fix $k = k(n) = \lfloor \frac{1}{\sigma \bar \chi} \ln q_n  - 1 \rfloor$.
       By definition,
       \[ A^z_{k+j}(\theta) = A^z_k(\theta+j\omega)A^z_j(\theta)  \]
      and, as $A^{f,z}_{k}$ is an $\slr$-cocycle, we have
      \be\label{slrmax}
                \max_j\left\{  \left\|A^{f,z}_j(\theta)\right\|,\left\|A_{j+k}^{f,z}(\theta)\right\|\right \}  \geq\exp\left\{\frac12(1-\tau)k\ly\right\}
      \ee
      for $|z-E| < e^{-\tau k\ly(E)}$.
      By (\ref{PGpars}) we can choose $t$ so that
      \begin{equation}
       \frac{\sigma\bar\chi}{\zeta} < t< \frac{(1-\tau)\chi}{\delta\xi}.\label{parb5}
      \end{equation}
      Finally, let $M_k = e^{t k}$. The first inequality in
      (\ref{parb5}) and the 
      choice of $k$ implies
       that for sufficiently large $n$,
      $M_k^\zeta\geq q_n + q_{n-1} -1 +k$. By  (\ref{parb5}),
      $M_k^{-\zeta} < e^{-\tau k \bar\chi}$ so 
      we have, for $|z-E| < M_k^{-\zeta}$  
      \begin{equation}
        \max_{1\leq j \leq M_k^\zeta} \left\|A_j^{f,z}(\theta)\right\|^2 \geq 
        e^{(1-\tau)k\chi} = M_k^{\frac{(1-\tau)\chi}{t}} 
         > M_k^\delta.
      \end{equation}
      Let $T_n = M_{k(n)}$.
      Then, for every $E \in U$, there is some $n_E$ so that for $n > n_E$ 
      (\ref{maxcocycle1}) holds.
      This settles the general case.
      
      If $\ly$ is continuous, and $U$ is compact then, by Lemma \ref{SOlemma}
      $k_\tau(E)$ is uniform for all  $E \in U$, and therefore 
       $n_E = n_0$ is also uniform for $E \in U$.

      For the Diophantine case for
      sufficiently large $T>0$, let $k = t^{-1}\ln T$.
      For large $n$, $q_{n+1} < q_n^{1+\kappa}$
      so there exists $q_n$ so that $ T^{\frac{\zeta}{1+2\kappa}}      < 2q_n+k <  T^{\zeta} $.
      Let $M_k $ be chosen so that 
      \[T^{\frac{\zeta}{1+2\kappa}}<M_k^{\frac{\zeta}{1+\kappa}} 
             < 2q_n +k< M_k^\zeta \leq T^{\zeta}. \]
       By construction and (\ref{PGpars}),
      $\delta ({1+2\kappa}) < (1-\tau)\frac{\chi}{t}$.
      It follows that, given $E \in U$ 
      there is a $T_E < \infty$ so that for $T > T_E$
      and $|z - E| \leq  T^{-\zeta}$  
      \begin{equation}
	    \max_{1\leq j \leq T^\zeta} \|A_j(\theta,z)\|^2 \geq 
		   \max_{1\leq j \leq M_k^\zeta} \|A_j(\theta,z)\|^2 \geq
           e^{(1-\tau)k\chi}\geq M_k^{\frac{(1-\tau)\chi}{t}} 
           > M_k^{\delta({1+2\kappa})}> T^\delta.
      \end{equation}
      The second inequality is simply (\ref{growexp}), and the remaining 
      inequalities follow from parameter choices.
 
     Again, if $\ly$ is continuous and $U$ is compact, 
     by Lemma \ref{SOlemma} 
      $T_E = T_0$ can be chosen uniform for $E \in U$.
 
      To complete the proof, it remains to show the
      transfer matrices grow on comparable lengths
      in the positive and negative directions.
      Note that for an ergodic invertible cocycle, the Lyapunov exponent of the forward cocycles 
      equals the Lyapunov exponent of the backward cocycles.
      Moreover, if $A_k^\omega$ is the
      cocycle over rotations by $\omega$, then the relation 
      $A_{-k}^{z,\omega}(\theta) = A_{k}^{z,-\omega}(\theta + \omega)$
      holds.
      Since $\omega$ and $-\omega$ have the same sequence of
      denominators $q_n$ from the continued fraction approximants, 
      we have that for $k$ large, $M_k$ may be 
      chosen exactly the same
      for $A^z_{k}$ and $A^z_{-k}$.\qed
      \end{proof}

   We will  obtain approximating polynomials for the rough potentials using
   Fejer's summability kernel
       \begin{equation}\label{fejer}
          K_N(\theta) = \frac{1}{N+1} \left(
           \frac{\sin\left(\frac{N+1}{2}\theta\right)}{\sin\left(\frac{1}{2}\theta\right)}
           \right)^2 
           = \sum_{-N\leq j\leq N} \left(1 - \frac{j}{N+1}\right)e^{ij\theta}. 
       \end{equation}
  Let $ \hat f(i)$ be the Fourier coefficients of $f$.  We have
  \[\sigma_N(f)(\theta) := K_N * f(\theta) = 
           \sum_{-N\leq j\leq N} 
           \left(1 - \frac{|j|}{N+1}\right)\hat{f}(j)e^{ij\theta}, \]
    is a $2N+1$st degree trigonometric polynomial.
   Moreover,   from the general theory, for $f\in L^1(\bbT)$, $\sigma_n(f)\to f$ in 
   $L^1(\bbT)$. The following is another standard result on the pointwise
   rate of convergence at well behaved points.

   The $\gamma$-Lipschitz function space $L_\gamma(\bbT)$ is 
   defined as the set functions on $\bbT$ with the norm,
   \[ \| f\|_{ L_\gamma} = \|f\|_\infty +
      \sup_{t\in\bbT; |h| > 0}\frac{|f(t+h) - f(t)|}{|h|^\gamma}. \] 
\begin{lemma}\footnote{This is a formulation of a standard result from harmonic analysis, see for example chapter 1 of \cite{K}.}\label{fla}
      Suppose $f\in L_\gamma(\bbT)$ 
      and for $\theta\in\bbT$ and $n\in\bbN$ we have
      \[ \left| K_n*f(\theta) - f(\theta)\right| < K \|f\|_{L_\gamma} n^{-\gamma} \]
      where $K$ does not depend on $n$.
   \end{lemma}
   \begin{proof}
   Observe $K_n$ has the following property,
   \be |K_n(\theta)| \leq \min\left\{n+1,\frac{\pi^2}{(n+1)\theta^2} \right\}. \label{FKbd}\ee
      Assume $f$ is $\gamma$-Lipschitz on $\bbT$ with constant $C$, then using (\ref{FKbd}) and $\sigma = \frac{1}{(n+1)}$
      \begin{eqnarray}
         \left| K_n*f(\theta) - f(\theta)\right| &=&\nonumber
              \left|\int_\bbT K_n(\tau)(f(\tau - \theta) - f(\theta))\de\tau \right| \\
         & \leq &\nonumber  \int_{[0,\pi]} \left|K_n(\tau)\right| 2C\tau^\gamma\de\tau  \\ 
         & \leq &\nonumber 2C \int_{[0,\sigma)}(n+1)\tau^\gamma\de\tau
             + 2C\int_{[\sigma, \pi)}\frac{\pi}{n+1}\tau^{\gamma - 2}\de\tau \\
         & \leq &\nonumber  CK'_\gamma n^{-\gamma}
      \end{eqnarray}
      Here $K$ does not depend on $f$ or $n$, and $C$ is the Lipschitz constant at $\theta$.\qed\end{proof}
   Let $\bbI$ be the set of intervals in $\bbT.$ 
   We say $f\in \bbI * L_\gamma(\bbT)$, if for $f_i \in L_\gamma(\bbT)$  and $I_i \in \bbI $ for $i = 1,\ldots,r$.
   \[   f(\theta) = \sum_{i=1}^r  1_{I_i} f_i(\theta).  \]
   \begin{lemma}
    \[  \bbI*L_\gamma(\bbT) = PL_\gamma(\bbT)\]
   \end{lemma}

   \begin{proof}
   One inclusion is clear. Suppose $f\in PL_\gamma(\bbT)$ where $f$ is continuous on $\bbT\backslash J_f$ for $\infty > |J_f| \geq 2$ 
   (if $f$ is continous everywhere there is nothing to show, if $f$ is discontinuous at only one point $x$ add $x+\pi$ to $J_f$).
   Let $I_i = (a_i,b_i)$ for $1\leq i \leq |J_f|$ be largest intervals in $\bbT\backslash J_f$ so that $\cup_i I_i = \bbT\backslash J_f$.
   The Lipschitz conditions ensure that
   limits $\lim_{\epsilon\to 0^+} f(a_i + \epsilon) = f(a_i+0) $  and $\lim_{\epsilon\to 0^+} f(b_i - \epsilon) =f(b_i-0) $ exist. 
   Now define $f_i$ to be equal to $f$ on $I_i$ and linearly interpolate the points $(b_i,f(b_i-0))$ and $(a_i,f(a_i+0))$
    on $I_i^c$, which clearly defines a $\gamma$-Lipschitz function.
   Then $f =  \sum_{i=1}^r  1_{I_i} f_i(\theta)$ so $f \in \bbI * L_\gamma(\bbT)$.\qed
   \end{proof}

  We will now show uniform upper bounds for cocycles in a neighborhood of $f
  \in  \bbI * L_\gamma(\bbT) .$ 

  \begin{lemma}\label{SOlemma2}
     Suppose $f\in  PL_\gamma(\bbT)$, and $E\in\bbC$ so that $\ly(E)>0$.
     For any $0<\tau<\|f\|_\infty^{-1}$ there exists a $k_\tau  =  k_{\tau}(E)< \infty$  so that if $q_n > e^{k_\tau\tau\ly(E)/\gamma}$ then for any $k\in\bbZ^+$ such that
     $k_\tau<k<\frac{\gamma}{\tau\ly(E)}\ln q_{n}$
    and 
     any $\theta \in \bbT$
     there is some $0<x\leq q_n + q_{n - 1} - 1$  so that for
     $z\in\bbC$ with $|z-E|<\exp\{-\tau k\ly(E)\}$ and $g\in
     \mathbb{B}_\infty(\bbT)$ with  $\|g-f\|_\infty <e^{-\tau
       k\ly(E)},\; \|g\|_\infty<\tau^{-1},$ we have
     \[ \left\|A_k^{g,z}\left(R^x\theta\right)\right\|\geq e^{k(1-\tau)\ly(E)}. \]
   \end{lemma}
   Lemma \ref{SOlemma} follows immediately from Lemma \ref{SOlemma2} with $g = f$.
   \begin{proof}
        It is clearly enough to prove the Lemma for $\tau<1.$ To begin we first fix some parameters for the proof.
         Let  
         \be \label{p4}
         \tau/2>\nu>\tau/4, \;\;\textrm{ and } \;\; 1-\tau/16>a>b>c>1-\tau/8  .
         \ee
         Finally, let $\bar\eta> 0$ be so small that $\bar\eta < \ly(E){\tau}/16$.

       Write $f = f_1 1_{I_1} +\cdots + f_r 1_{I_r}$ for Lipschitz functions $f_i\in L_\gamma(\bbT)$  and intervals $I_i$. 
       There is no loss of generality if we assume $r \geq 2$.
       Let $J(1_{I_i})$ be the set of  discontinuities  of $1_{I_i};$ 
       then the set of discontinuities of $f$ is $J_f = J(f) = \cup_{i=1}^r J(1_{I_i})$.
      In practice,  we will use a simple bound for the supremum norm of $f$
       \[  \|f\|_\infty \leq   \|f_1\|_\infty +\cdots +\|f_r\|_\infty=: M. \]
       Observe, for $h \in L_\infty(\bbT)$,
       we have 
 $\|K_N*h\|_\infty \leq \|h\|_\infty$.
      For $f\in  \bbI* L_\gamma(\bbT)$, we write $f_N = \sigma_N(f_1)1_{I_1}+ \cdots + \sigma_N(f_r) 1_{I_r}$,
       so we have $\|f_N\|_\infty \leq M$.
       It is clear that
      \[\|A^{h,E}\|_\infty \leq  1+  \|h\|_\infty+ |E|,\]
      so we easily have uniform bounds for the cocycle matrices over bounded energies and unifomly bounded potentials.

       Let $\epsilon >0$.
       There is some $\rho_\epsilon >0$ and $K_\epsilon<\infty$  (depending on $E$)
      so that for $k > K_\epsilon$, and $|z-E| + \|g-f\|_{\infty}< \rho <
      \rho_\epsilon$  we have, from Corollary \ref{mc} with $\eta =1$,   
      \be\|A^{f,E}_k(\theta)-{A}^{g,z}_k(\theta)\| <\label{mc1}
           \rho e^{k(\ly(E) + \epsilon \overline M)} \ee
       where  
$\overline M =\max\{1,\ln[ 1 + M + \rho_\epsilon +|E|]\}  $.

       In particular if $g = f_N$,  we have by Lemma \ref{fla}
       \be \label{fejerpotential}
         \|f_N(R^i\theta) - f(R^i\theta)\|_\infty < C_f N^{-\gamma}   ,
       \ee
       Set
       \[ N = \exp\left\{ \ly(E) k\frac{\nu}{\gamma} \right\}\]
       and let  $\epsilon < {\bar\eta}/{\overline M} $.
       Now there is some $K_\epsilon$ so that we can apply (\ref{mc1}), which we will call $M_\epsilon$.
       That is, we have
      $C_fN^{-\gamma}< \tfrac12\rho_\epsilon$  and for $|z-E|< C_fN^{-\gamma}$ we have 
    \be\|A^{f,E}_k(\theta)-{A}^{f_N,z}_k(\theta)\| <\label{mc2}
           C_f N^{-\gamma} e^{k(\ly(E) + \eta )}. \ee

       Let $A^{f_N,E}$ be the cocycle matrix defined by 
       the potential determined by the sampling function $f_N , $ 
        which is a piecewise polynomial on $|J_f|$ intervals, each supporting a continuous polynomial of order $(2N+1)$.

        For a map $B:\bbT\to SL_2(\bbR)$ and associated cocycle set
        \begin{equation}\label{vsets}
	    V_k\left(t,B\right) = 
		\left\{\theta\in \bbT :\frac{1}{k}\ln\|B_k(\theta)\|
                      > t\right\}\subset\bbT.
        \end{equation}
        The measure of this set for $B = A^{f,E}$ for large enough $k$ 
        can be bounded below using the fact that
        $\ly(E) = \inf_k \int_\bbT \frac{1}{k}\ln\|A^{f,E}_k(\theta)\|\de\theta$.
        Indeed, by Corollary \ref{pwfurmana} there is $k_{f,E }< \infty$ so that for $k > k_{f,E}$
         we have for all $\theta$,
        $\frac{1}{k}\ln\|A^{f,E}_k(\theta)\| < \ly(E)+\bar\eta$,
        thus,
        \begin{eqnarray}
	     \ly(E) &\leq& \nonumber \int_\bbT \frac{1}{k}\ln\|A^{f,E}_k(\theta)\|\de\theta \\
              & \leq &\nonumber
                  |V_k\left(a\ly(E),A^{f,E}\right) |(\ly(E)+\bar\eta)  + \left|V_k^c\left(a\ly(E),A^{f,E}\right)\right|a\ly(E)\\  
              & \leq &\nonumber
                  |V_k\left(a\ly(E),A^{f,E}\right) |  [(1-a)\ly(E)+\bar\eta] +a  \ly(E) .  
        \end{eqnarray}
        Note that if $U$ is compact and $\ly$ is continuous then, by Lemma \ref{cuLyapunov},
        $k_{f,E}$ can be chosen uniformly for $E\in U$.
         By the choice of $\bar\eta$ we have $\bar\eta < ( 1 - a) \ly(E)$ so 
for $k > k_f$,
        \be \label{lebound} \frac{1}{2} \leq \frac{ ( 1 - a )\ly(E) }{ (1-a)\ly(E)+\bar\eta } \leq |V_k\left(a\ly(E),A^{f,E}\right) |  . \ee
        
        Furthermore, we make the following claim regarding the sets $V_k(\cdot,\cdot)$,
        there is some $k_\tau(E) < \infty$ so that for $k > k_\tau(E)$ and
        $|E - {z}| < \exp\{- \ly(E) \tau k\}$,
        \begin{equation}\label{incls}
           V_k(a\ly(E),A^{f,E}) 
            \subset V_k(b\ly(E),A^{f_N,E}) 
            \subset V_k(c\ly(E),A^{g,{z}}).
        \end{equation}
        First note from the assumption on parameters (\ref{p4}),
       \be\label{cparameter}
        \ly(E)(1-\nu) + \bar\eta < \ly(E)(1-\tau/4) + \tau\ly(E)/16 = \ly(E)(1- 3\tau/16) < \ly(E) c 
        \ee
      to show the left inclusion, for $ \theta\in V_k \left(a\ly(E),A^{f,E}\right)$ write
        \[  \left\|A^{f_N,E}_k(\theta)\right\|
	     \geq
                 \left\|A^{f,E}_k(\theta)\right\| - \left\| A^{f_N,E}_k(\theta) -  A^{f,E}_k(\theta)\right\| \]
               from  (\ref{mc1}) and (\ref{fejerpotential}) we have
            \[  \left\| A^{f_N,E}_k(\theta) -  A^{f,E}_k(\theta)\right\|
                  \leq C \cdot N^{-\gamma} e^{k(\ly(E)+\bar\eta)} \leq  Ce^{k(\ly(E)+\bar\eta - \ly(E)\nu)}   
                   <  Ce^{ ck \ly} \]
            having  used the definition of $N$ in the second to last step and 
            (\ref{cparameter}) in the last step. Putting this together,
             we have, 
            \[ \left\|A_k^{f_N,E}(\theta)\right\|
	    >
              e^{ak\ly(E)} - Ce^{ c k \ly(E)}    >  e^{bk\ly(E)}.\]
              Where the final inequality clearly holds for large enough $k$.
          The right inclusion of (\ref{incls}) is similar: for $ \theta\in V_k \left(b\ly(E),A^{f_N,E}\right)  $
	\[
             \left\|A^{g,{z}}_k(\theta)\right\|
             > 
               \left\| A_k^{f_N,E}(\theta) \right\| -  \left\| A_k^{f_N,E}(\theta)  - A_k^{f,E}(\theta) \right\|   -  \left\|A_k^{f,E}(\theta)- A_k^{g,z}(\theta) \right\|\]
         The second term on the right can be bounded as above, the last term on the right is bounded similarly 
         using (\ref{mc1}) and $\epsilon < \bar \eta /\overline{M}$ so that
         $ \left\|A_k^{f,E}(\theta)- A_k^{g,z}(\theta) \right\| 
               <  Ce^{ ck \ly} $.
          Thus, as $b > c$, for large enough $k$ we have,  
         \[   \left\|A^{g,{z}}_k(\theta)\right\| > 
              e^{bk\ly(E)} -  
              Ce^{ck\ly(E)}
             >  e^{ck\ly(E)}\]
              which proves the claim.
              
              As before, note that if $\ly$ is continuous, then $k_\tau(E)$ can be chosen uniform for $E$
              in a compact set $U$, since the only possible source of 
              nonuniformity over $E$ is the requirement that $k > k_{f,E}$.
              
              Now write 
           $V 
            = V_{k}(b\ly(E),A^{f_N,E })$,
          combining (\ref{incls}) and (\ref{lebound}) yields
           $  \left|V\right| 
              \geq \frac{1}{2}.$
         On the other hand, $V$ is defined by a piecewise polynomial function. 
          That is, $\bbT$ is partitioned into $k|J|$ intervals and on each interval
             $\|A^{f_N,E}_k(\theta)\|^2 $ is a polynomial of degree $2k(2N+1) $.
          At least one interval in the partition must have an intersection with $V$ of size $\frac{1}{4}(k|J|)^{-1}$,
          and therefore $V$ must contain 
an interval of length $\frac{1}{4}\frac{1}{k^2(2N+1)|J|}$,
          which is bounded below by $\exp\{- k \tfrac{\tau}{\gamma} \ly(E)\}.$
          It follows from (\ref{incls}) that $V_k(c\ly(E),A^{g,z}) $
	  also contains this interval. 
          We will now use the following fact:

   \begin{lemma}\label{JitoLemma}(e.g. \cite{JL00})
       For an interval $I\subset\bbt$, if $n$ is such that 
       $|I| > \frac{1}{q_n}$ then for any $\theta\in\bbt$ there is 
       $0\leq j \leq q_n + q_{n-1} - 1$ so that $\theta + j\omega \in I$.
   \end{lemma}

          For any $ k_\tau < k \leq \frac{ \gamma\ln q_n}{\tau\ly }$ we have
          that $V_k(c\ly(E),A^{g,z}) $ contains an interval of length greater than
           $\exp\{- k \tfrac{\tau}{\gamma} \ly(E)\}$ which in turn is greater than   $ \frac{1}{q_n}$, so for
          some $0\leq x \leq q_n+q_{n-1}-1$ we obtain the result.
	  \qed
          \end{proof}

\section{Acknowledgement} 
    We would like to thank Rui Han for helpful input on an early draft
    and the anonymous referee for careful reading of the
    manuscript that has led to an important improvement.
    S.J. is a 2014-15 Simons Fellow. This research was partially
    supported by NSF DMS-1101578 and  DMS-1401204.

\bibliographystyle{plain}	
\bibliography{noterefs}

\end{document}